\documentclass[12pt]{article}
\usepackage{amsmath,amsfonts,latexsym, amssymb}
\usepackage{amsthm,amsmath,amsfonts,amssymb,cases}
\newcommand{\R}{{\mathord{\mathbb R}}}
\newcommand{\Z}{{\mathord{\mathbb Z}}}
\newcommand{\N}{{\mathord{\mathbb N}}}
\newcommand{\C}{{\mathord{\mathbb C}}}

\def\chib {\overline{\chi}}

\newcommand{\HH}{\mathcal{H}}

\newcommand{\FF}{\mathcal{F}}

\newcommand{\VV}{\mathcal{V}}
\newcommand{\WW}{\mathcal{W}}
\newcommand{\hh}{\mathfrak{h}}
\newcommand{\vv}{\mathfrak{v}}
\newcommand{\UU}{\mathcal{U}}

\newcommand{\ran}{{\rm Ran}}

\newcommand{\ben}{\begin{displaymath}}
\newcommand{\een}{\end{displaymath}}
\newcommand{\beqn}{\begin{equation}}
\newcommand{\eeqn}{\end{equation}}
\newcommand{\beqna}{\begin{eqnarray*}}
\newcommand{\eeqna}{\end{eqnarray*}}

\newcommand{\be}{\begin{equation}}
\newcommand{\ee}{\end{equation}}

\def\inf{{\rm inf}\,}

\newtheorem{theorem}{Theorem}
\newtheorem{corollary}[theorem]{Corollary}
\newtheorem{lemma}[theorem]{Lemma}
\newtheorem{proposition}[theorem]{Proposition}

\newtheorem{remark}{Remark}
\newtheorem{definition}{Definition}

\numberwithin{equation}{section}
\numberwithin{theorem}{section}
\numberwithin{definition}{section}

\begin{document}

\title{A vanishing theorem for operators in Fock space}
\author{\vspace{5pt} D. Hasler $^1$\footnote{
E-mail: david.hasler@math.lmu.de, on leave from Ludwig Maximilians University} and I.
Herbst$^2$\footnote{E-mail: iwh@virginia.edu.} \\
\vspace{-4pt} \small{$1.$ 
Institute of Mathematics, University of Munich, } \\
\small{  D-80333 Munich, Germany     }\\
\vspace{-4pt}
\small{$2.$ Department of Mathematics, University of Virginia,} \\
\small{Charlottesville, VA, 
 USA}\\}
\date{}
\maketitle

\begin{abstract}
We consider the bosonic Fock space over the Hilbert space of transversal vector fields in three dimensions.
This space carries a canonical representation of the group of rotations.
For a certain class of operators in Fock space we show that rotation invariance implies
the absence of terms which either create or annihilate only a single particle.
We outline an application of this result in an operator theoretic
renormalization analysis of Hamilton operators, which occur  in non-relativistic qed.
\end{abstract}

\medskip


\section{Introduction}

In this paper  we consider the bosonic Fock space over the Hilbert space of transversal
vector fields in three dimensions. This space is used in mathematical  models of   quantized radiation and  carries a canonical representation of the group of rotations.
We consider a class of   operators in this Fock space  which   arise in a so called
operator theoretic renormalization
analysis of non-relativistic quantum electrodynamics (qed) \cite{BFS98,BCFS03}.
For operators in this class we prove  that  rotation invariance implies
the absence of terms which either create or annihilate only a single particle.
This vanishing theorem implies that   under  a non-degeneracy assumption
marginal terms in operator theoretic renormalization  are  absent.
In \cite{HH10-2,HH10-3} this property  was used to obtain
ground state  properties in  non-relativistic qed, such as analyticity in a
minimal coupling constant or differentiability in the fine structure constant.

The vanishing theorem  was first shown
in  \cite{HH10-2}, where  the  proof relied on a  uniqueness result from \cite{BCFS03} (Theorem 3.3).
The proof  in the present paper is self-contained.
Since vanishing theorems have far reaching consequences in the
context of operator theoretic renormalization, we consider the
short and new proof presented in this paper of value to the scientific literature.

In Section~\ref{sec:modsta} we introduce the Fock space and define a canonical representation
of the group of rotations in three dimensions.
In  Theorem~\ref{thm:main} the vanishing theorem is stated, the main result of this paper.
In Subsection~\ref{sec:app} we introduce the  Hamiltonian of non-relativistic qed
and show that it is rotation invariant, with respect to the representation defined
in Section~\ref{sec:modsta}.
In Subsection~\ref{sec:smo},   the vanishing theorem is used to derive a corollary, which
outlines applications to operator theoretic renormalization. 

\section{Model and Statement of Main Result}
\label{sec:modsta}

We consider  the special group of orthogonal matrices  in three dimensions
$$
G := \{ R \in M_3(\R) \  | \ {\rm det} R = 1 , \, R^T = R^{-1} \} ,
$$
where $M_3(\R)$ denotes the set of $3 \times 3$ matrices over the real
numbers, with the usual topology. As a subset of  $M_3(\R)$ the group $G$ inherits
a natural topology.

\begin{definition}
A representation of $G$ ($G$--representation) is a strongly continuous map $\UU : G \to \mathcal{B}(\HH)$ to the unitary
operators on $\HH$, such that
$$
\UU(R_1 R_2) = \UU(R_1) \UU(R_2) ,  \quad \forall R_1, R_2  \in G .
$$
\end{definition}

We shall adopt the following standard conventions. Given a representation $\UU$ of
$G$ on $\HH$,
a vector $v \in \HH$  is called $\UU$--invariant if $\UU(R) v = v$  for all $R \in G$.
A subspace $V \subset \HH$ is called $\UU$--invariant if $\UU(G) V \subset V$.
An operator  $T$ in $\HH$ with domain $D$ is  called $\UU$--invariant if
$\UU(R) D = D$ and   $\UU(R) T \UU(R)^* = T$ for all $R \in G$.
We will write $G$--invariant instead of $\UU$ invariant, if it is clear from the context what the representation is. By  rotation invariant we shall always mean $G$--invariant.

The Hilbert space of vector fields  $\VV=L^2( \R^3 ; \C^3)$
carries a natural  representation of $G$, which will
   be  denoted by   $\mathcal{U}_{\VV}$. Explicitly on vector fields  $v \in \VV$  it acts as
\begin{equation} \label{eq:defofvec}
 \left[ \,  \mathcal{U}_\VV(R) v \, \right] (\cdot ) = R  \left[ v (R^{-1} \, \cdot \,  ) \right] ,
\end{equation}
for   all $R \in G$.
We define the   subspace
$$
\vv := \{ \ v \in \VV \ | \ k \cdot v(k)  = 0 \ {\rm a.e.} \ k \in \R^3 \  \} ,
$$
 of transversal vector fields in $\VV$.
It is straightforward to verify, that
 $\vv$ is invariant under the representation $\UU_\VV$.
The main result will be based on the following  lemma.

\begin{lemma} \label{lem:onepart}
If  $v \in \vv$ is $G$--invariant, then $v=0$.
\end{lemma}
We note that under  the simplifying assumption that $v \in \vv$ is
defined everywhere and $(\UU_\vv(R) v)(k) = v(k)$ for all $k \in \R^3$,
the assertion of the Lemma follows trivially from the hairy ball
theorem. To this end,  observe that for such a $v \in \vv$
\be \label{eq:ginv}
  v(  r  R e  ) = R v(   r e  ) ,
\ee
for all $R \in G$, $ r \geq 0$, and $e \in  S_2 := \{ e   \in \R^3 | \ |e   | = 1 \}$.
Eq. \eqref{eq:ginv} implies that for every $r > 0$, the function
$e \mapsto v( r e)$ is a continuous function of   $e \in S^2$ and its value has
constant norm.  On the other hand $v( r e)$ is tangential to $S^2$ and this contradicts the hairy ball theorem unless it  is the zero vector.
The problem with this argument is that elements of
$\vv$ are only defined up to sets of measure zero.
To deal with this issue we  decompose $\vv$ into  the irreducible representations  of $G$.
The irreducible unitary representations of $G$ are denoted  by $D_j$,
$j \in \N_0$, and are  uniquely determined up to
unitary equivalence by their dimension ${\rm dim} \, D_j = 2 j + 1$. For the proof of Lemma \ref{lem:onepart} we use the following  idea from \cite{HH10-2}.
\medskip

\noindent
{\it Proof of Lemma \ref{lem:onepart}}.
By means of the  canonical isomorphism
\begin{equation} \label{lem:onepartttt}
 L^2(\R_+ ; r^2 dr ) \otimes L^2(S_2  ; \C^3)  \cong   \VV
\end{equation}
we can identify
$$
\VV_0 := L^2(\R_+; r^2 dr ) \otimes \mathcal{S} ,
$$
where $\mathcal{S} := \{ f \in L^2(  S_2 ; \C^3) \   | \ f(e) = \lambda e, \,  \lambda \in \C \}$,
with  a $G$--invariant subspace of $\VV$. It follows from \eqref{eq:defofvec}
that each element in $\VV_0$ is $G$--invariant. On the other hand it is an
immediate consequence of the definition of $\mathcal{S}$ that $\VV_0$ does not contain
any nonzero transversal vector fields. The Lemma will follow if we can show that every
$G$--invariant element of $\VV$ lies in $\VV_0$. To this end, let  $H_l$ denote the space of spherical harmonics with angular momentum $l$.  Using
$L^2(S_2 ; \C^3 ) = \bigoplus_{l=0}^\infty H_l \otimes \C^3 $, we find from \eqref{lem:onepartttt}
 an isomorphism of Hilbert spaces
\begin{align}
\VV   
& \cong  L^2(\R_+; r^2 dr ) \otimes  \bigoplus_{l=0}^\infty \left\{  H_l  \otimes \C^3  \right\}  .
\label{eq:rep3}
\end{align}
From  \eqref{eq:defofvec} it follows that  $G$ acts on
$H_l \otimes \C^3$ as $D_l \otimes D_1$. By the Clebsch--Gordon decomposition
we see that this representation contains the trivial representation only
if $l=1$, in which case
$D_1 \otimes D_1 \cong D_2 \oplus D_1 \oplus D_0$. Since  the
trivial representation only occurs
once in $L^2(S_2 ; \C^3)$,  it follows that $\VV_0$ contains
all elements in $\VV$ which are $G$--invariant.
 \qed

\medskip

We choose two Borel-measurable endomorphisms
$$\varepsilon_1, \, \varepsilon_2  : S_2   \to S_2$$
with the  property that for a.e. $e \in S^2$,
$
\varepsilon_1( e )  \times \varepsilon_2( e ) = e  ,
$
where  $\times$ denotes the vector product.
We extend these mappings  to the set $\R^3_\times :=  \R^3 \setminus \{ 0 \}$
by setting
$\varepsilon_\lambda(k)  = \varepsilon_\lambda( k/|k|)$ for
$\lambda=1,2$ and $k \in \R^3_\times$.
The explicit choice of these so called  polarization vectors
establishes a canonical isomorphism,  $\phi$, from the
Hilbert space
$\hh := L^2( \R^3 \times \Z_2 )$, to the Hilbert space of  transversal
vector fields,
i.e.,
\begin{eqnarray}
\phi : \ \hh  &\to& \vv   \nonumber  \\
  h &\mapsto&   \Big\{  k \mapsto   \varepsilon_1(k) h(k,1)
+  \varepsilon_2(k) h(k,2)  \Big\} . \label{eq:defofphi}
\end{eqnarray}
By means of this  isomorphism we obtain the $G$--representation, $\UU_\hh$,  on $\hh$ as follows
\begin{equation} \label{eq:defofUhh}
 \UU_\hh(R)  = \phi^{-1} \   \UU_\vv(R) \ \phi  ,\quad \forall R \in G .
\end{equation}

From Eq. \eqref{eq:defofUhh} the following   Corollary is an immediate consequence of  Lemma  \ref{lem:onepart}.

\begin{corollary} \label{lem:onepart}
If  $h \in \hh$ is $G$--invariant, then $h=0$.
\end{corollary}

Next we introduce the bosonic Fock space over the Hilbert space $\hh$. We define
$$
\FF  := \bigoplus_{n=0}^\infty  \FF_n
$$
with $\FF_0 := \C$ and   $\FF_n :=  S_n ( \hh^{\otimes n} )$ for  $n \in \N$,
where $S_n$ denotes the orthogonal projection onto the subspace of totally symmetric
tensors in $ \hh^{\otimes n}  $.
The vector $\Omega := (1,0,0,...) \in \FF$ is called the Fock vacuum.
The Fock space inherits a natural inner product from the Hilbert space $\hh$.
The creation operator, $a^*(f)$, for $f \in \hh$
 is defined on vectors $\eta \in  \FF_n$, with  $n \in \N_0$,
by
\beqn \label{eq:formala}
a^*(f) \eta := (n+1)^{1/2} S_{n+1} ( f \otimes \eta ) \; .
\eeqn
We extend this  definition  by linearity to  a  densely defined linear operator on $\FF$.
The resulting operator  is  closable and its closure will be denoted by the same symbol.
We introduce the annihilation operator by
\begin{equation}  \label{eq:defofanihi}
a(f) := \big\{ a^*(f) \big\}^*.
\end{equation}

For  a bounded linear operator, $A$, on  $\hh$, we denote by $\Gamma(A)$   the unique bounded linear  operator on $\FF$ satisfying $\Gamma(A) \Omega = \Omega$ and
$$
\Gamma(A)  S_n ( \varphi_1 \otimes \cdots \otimes \varphi_n )  = S_n ( A  \varphi_1 \otimes \cdots \otimes  A \varphi_n  ) ,
$$
for any $\varphi_1,...,\varphi_n \in \hh$ and $n \in \N$.
The Fock space $\FF$
carries a natural representation of $G$ given by
\begin{equation} \label{eq:defofuu}
\UU_\FF := \Gamma(\UU_\hh) .
\end{equation}

Below we introduce operators on Fock space, which arise in operator theoretic renormalization.
First we introduce the operator of the free field energy. We  define   the function
 $\omega : \R^3 \times \Z_2 \to \R  , ( k, \lambda) \mapsto |k|$ and denote
the corresponding multiplication operator on $\hh$ by the same symbol.
  The operator
 of the free field energy, $H_f$,  is defined as the
unique selfadjoint operator on Fock space such that for all $t \in \R$
\begin{equation} \label{eq:defofH}
 e^{- i H_f t } = \Gamma(e^{- i \omega t} )  .
\end{equation}
We introduce the sets
 $$ B_1 := \{ k \in \R^3 | |k| \leq 1 \} , \quad  X :=  B_1  \times \Z_2  , \quad  I:=[0,1] . $$
For $m,n \in \N_0$, we define   the Banach space, $\WW_{m,n}$, with norm $\| \cdot \|_{\WW_{m,n}}$, defined  in \eqref{eq:defofnorm} below, to consist of the measurable functions
$$
 w_{m,n} :  X^{m} \times X^{n} \to C^0(I; \C ) ,
$$
satisfying the symmetry property \eqref{eq:sym} and  the support
property  \eqref{eq:supportprop}, below.
The symmetry property states that for all
$\tilde{K}_1,...,\tilde{K}_m, K_1,..., {K}_n \in X$, one has
\begin{equation} \label{eq:sym}
w_{m,n}((\tilde{K}_1,...,\tilde{K}_m), (K_1,...,K_n)  ) = w_{m,n}( (\tilde{K}_{\tilde{\sigma}(1)},...,\tilde{K}_{\tilde{\sigma} (m)}), (K_{\sigma(1)},...,K_{\sigma(n)}) ) ,
\end{equation}
for any permutation $\sigma$ and  $\tilde{\sigma}$  of $\{1,...,n\}$ and
$\{1,...,m\}$, respectively.
For simplicity of notation we shall  write
$$ w_{m,n}(r ;K^{(m,n)}) =  w_{m,n}(K^{(m,n)})(r),$$
for  $r \in I$  and $K^{(m,n)} \in X^{m} \times X^n$.
Moreover, we  introduce the following notations
\begin{align}   \label{eq:defofK}
& K^{(m,n)} := (\tilde{K}^{(m)},K^{(n)}) :=  ((\tilde{k}_1, \tilde{\lambda}_1,...,\tilde{k}_m, \tilde{\lambda}_m ), ( {k}_1, {\lambda}_1,...,{k}_n, {\lambda}_n)) \in X^m \times X^n  \\
& |K^{(m)}| := \prod_{j=1}^m |k_j| ,
\ \   d K^{(m)} := \prod_{j=1}^m d k_j^3 , \ \ \Sigma[{K}^{(n)}]  := \sum_{j=1}^n | k_j |   \ .\label{eq:defofK2}
\end{align}
The support property states that
\begin{equation}
\label{eq:supportprop} w_{m,n}(r ; K^{(m,n)})  =
1_{ \Sigma[\tilde{K}^{(m)}] + r \leq  1}  w_{m,n}(r ; K^{(m,n)}) 1_{\Sigma[{K}^{(n)}] + r \leq  1} ,
\end{equation}
for all $r \in I$ and $K^{(m,n)} \in X^{m+n}$.
The norm is given by
\be \label{eq:defofnorm}
\| w_{m,n} \|_{\WW_{m,n}} :=  \left\{ \int_{ X^m \times X^n  }
{\rm sup}_{r \in I} | w_{m,n}(r; K^{(m,n)} )  |^2
d \mu_{m,n}^{(2)}(K^{(m,n)})  \right\}^{1/2}  ,
\ee
with measure
$$
d \mu_{m,n}^{(p)}(K^{(m,n)}) := \frac{ d \tilde{K}^{(m)}}{|\tilde{K}^{(m)}|^p} \frac{ d K^{(m)}}{|K^{(n)}|^p}  .
  $$
The integration in \eqref{eq:defofnorm} includes 
summation over $\lambda_i$ and $\tilde{\lambda}_i$.
For  $0 < \xi < 1$,  we define the Banach space $\WW_\xi$
to consist of sequences
$\underline{w} = (w_{m,n})_{(m,n) \in \N_0^2}$ with $w_{m,n} \in \WW_{m,n}$
such that the norm
\be  \label{eq:defofxi}
\| \underline{w}\|_\xi := \sum_{m+n \geq 0} \xi^{-(m+n)} \| w_{m,n} \|_{\WW_{m,n}}
\ee
is finite.
Let  $\underline{w} \in \WW_\xi$.  We define the operator
\begin{eqnarray}
 H_{m,n}[\underline{w}] :=    \label{defofintop}
 \int_{X^m \times X^n  } \prod_{i=1}^m
 a^*_{\tilde{\lambda}_i}(\tilde{k}_i)     w_{m,n}(H_f ; K^{(m,n)} )
\prod_{j=1}^n a_{{\lambda}_j}({k}_j)  d \mu_{m,n}^{(1/2)}(K^{(m,n)} )  .
\end{eqnarray}
 Definition \eqref{defofintop} is understood in the sense of forms  and the right hand side
of  \eqref{defofintop} as a weak integral. Moreover, the definition involves the
operator valued distributions $a_\lambda(k)$ and $a_\lambda^*(k)$ which are
characterized by the relations
$$
a^*(f) =  \sum_{\lambda=1,2}  \int f(k,\lambda) a^*_\lambda(k) d^3 k  , \quad a(f) = \sum_{\lambda=1,2} \int \overline{f(k,\lambda)} a_\lambda(k)  d^3 k ,
$$
for all $f \in \hh$.
In the Appendix we give an explicit and rigorous  definition of
 \eqref{defofintop} which does not
involve any creation or annihilation operators.

\begin{remark} In this paper we
view   $H[\underline{w}]$ as  on operator on Fock space.
To this end we tacitly
extend the kernels $w_{m,n}$ to functions on $\R \times \{ \R^3 \times \Z_2 \}^{m+n}$
by setting them equal to zero on the complement of $I \times X^{m+n}$.
Alternatively, we could view  $H[\underline{w}]$  to be an operator
on $1_{H_f \leq 1} \FF$. This would not change the analysis.
By assumption \eqref{eq:supportprop} we can neglect projection operators $1_{H_f \leq 1}$
in the definition   \eqref{defofintop}, which appear   in
 \cite{BCFS03}.
\end{remark}

With respect to the operator norm in Fock space the following estimate is shown
in the Appendix,
\begin{equation} \label{eq:boundonH}
\| H_{m,n}[\underline{w}] \| \leq   \| w_{m,n} \|_{\WW_{m,n}} .
\end{equation}
We define the operator
\be \label{defofintop0}
H[\underline{w}] := \sum_{m+n \geq 0} H_{m,n}[\underline{w}]  .
\ee
It follows from \eqref{eq:boundonH} and the definition of the norm \eqref{eq:defofxi}
that the right hand side of \eqref{defofintop0} converges absolutely with respect to
the  operator norm, and furthermore,
$$
\| H[\underline{w}] \| \leq \| \underline{w} \|_\xi .
$$
Now we state the main result of this letter.
\begin{theorem} \label{thm:soinv} \label{thm:main}  Let $\underline{w} \in \WW_\xi$.
If $H[\underline{w}]$ is a $G$--invariant operator, then $H_{0,1}[\underline{w}]=0$ and $H_{1,0}[\underline{w}]=0$.
\end{theorem}

\begin{proof}
For  $f_1,f_2, g  \in \hh$ it follows  from  definition \eqref{defofintop}
(or respectively definitions  \eqref{eq:defofL2s} and  \eqref{eq:zero}), that
\begin{eqnarray}
\lefteqn{  \frac{1}{\sqrt{2}}\langle g    , H[\underline{w}] S_2 ( f_1 \otimes f_2 )  \rangle = }
  \label{eq:eq450}\\
&&  \frac{1}{2} \sum_{\lambda,\tilde{\lambda}} \int \overline{ g(\tilde{k},\tilde{\lambda})} f_{1}( \tilde{k}, \tilde{\lambda})
w_{0,1}(|\tilde{k}|; k, \lambda) f_2(k , \lambda)\frac{d^3 \tilde{k} d^3 k}{\sqrt{|k|}}
\label{eq:eq45} \\
&&  +  \frac{1}{2} \sum_{\tilde{\lambda},\lambda} \int \overline{ g(\tilde{k},\tilde{\lambda})} f_2( \tilde{k}, \tilde{\lambda})
w_{0,1}( |\tilde{k}| ; k, \lambda) f_{1}(k , \lambda)  \frac{d^3 \tilde{k} d^3 k}{\sqrt{|k|}}
  \label{eq:eq456} \\
&&  +  \sum_{\tilde{\lambda},\lambda_1,\lambda_2} \int \overline{ g(\tilde{k},\tilde{\lambda})}
w_{1,2}( 0 ;  \tilde{k}, \tilde{\lambda} , k_1, \lambda_1 , k_2, \lambda_2 )  f_{1}(k_1 , \lambda_1) f_2( k_2, \lambda_2)  \frac{ d^3 \tilde{k} ...  d^3 k_2}{\sqrt{|\tilde{k}||k_1||k_2|}}    .    \label{eq:eq4567}
\end{eqnarray}
Pick a function $\varphi \in C_0^\infty(B_1 ; \R) $ with $\int \varphi^2 (k) d^3 k  = 1 $,
and  for $x \in B_1^0 := \{ k \in \R^3 | |k| < 1\}$   define
\begin{align} \label{eq:defoffep}
 f_{\epsilon,x}(k,\lambda)  :=  2^{-1/4} {\epsilon^{-3/2}} \varphi( \epsilon^{-1} ( {x-k}) ) , \quad  \epsilon > 0.
\end{align}
Notice that  $f_{\epsilon,x}$ converges weakly to zero  in $L^2$  as $\epsilon \downarrow 0$.

We insert the choice  $f_1 = g = \UU_{\hh}(R) f_{\epsilon,x}$ and $f_2 = \UU_{\hh}(R) h$, with $h \in \hh$, into \eqref{eq:eq450}.
We claim that in  the limit $\epsilon \downarrow 0$, the terms in lines \eqref{eq:eq456}
and  \eqref{eq:eq4567} vanish. To this end note that
\begin{equation}
G_1(\tilde{k},k)  = f_2(\tilde{k},\tilde{\lambda})w_{0,1}(|\tilde{k}|,k,\lambda)/\sqrt{|k|}
\end{equation}
and
\begin{equation}
G_2(\tilde{k},k_1) = \int w_{1,2}(0,\tilde{k},\tilde{\lambda}, k_1,\lambda_1,k_2,\lambda_2)f_2(k_2.\lambda_2)\frac {d^3k_2}{\sqrt{|\tilde{k}||k_1||k_2|}}
\end{equation}
are kernels of Hilbert-Schmidt operators so by the weak convergence of $g = f_1$ to $0$, the terms in lines \eqref{eq:eq456}
and  \eqref{eq:eq4567} vanish.


An elementary calculation using the definition of the polarization vectors, the group of rotations,
and   \eqref{eq:defofUhh} shows that
$$
\sum_{\lambda=1,2} \left| (\UU_\hh f_{\epsilon , x })(k,\lambda)\right| =
 \sum_{\lambda=1,2} \left|  f_{\epsilon , x }(R^{-1}k,\lambda) \right|^2 , \quad \forall k \in \R^3 .
$$
Using this and the vanishing of  \eqref{eq:eq456}
and  \eqref{eq:eq4567} in the limit $\epsilon \downarrow 0$, we find
\begin{eqnarray}
\lefteqn{ \lim_{\epsilon \downarrow 0} \langle  \UU_{\hh}(R) f_{\epsilon,x}  , H[\underline{w}] S_2 ( \UU_{\hh}(R) f_{\epsilon,x}  \otimes  \UU_{\hh}(R) h ) \rangle  } \label{lhsinvR} \\
& = &  \lim_{\epsilon \downarrow 0} \frac{1}{\sqrt{2}} \sum_{\lambda,\tilde{\lambda}} \int   |f_{\epsilon,x}(R^{-1}\tilde{k},\tilde{\lambda})|^2
w_{0,1}(|\tilde{k}|; k, \lambda) ( \UU_\hh(R)  h)(k , \lambda)\frac{d^3 \tilde{k} d^3 k}{\sqrt{|k|}}  \label{lhsinvR000}  .
\end{eqnarray}
Moreover, using dominated convergence one finds
\begin{eqnarray}
\eqref{lhsinvR000} =  \frac{1}{\sqrt{2}} \sum_{\lambda,\tilde{\lambda}}  \int \lim_{\epsilon \downarrow 0} I(\epsilon, k,\lambda,\tilde{\lambda}, R,h) {d^3 k}  \label{dom:conv} ,
\end{eqnarray}
where we  introduced the notation
$$
I(\epsilon, k,\lambda,\tilde{\lambda}, R,h ) :=  \int  |f_{\epsilon,x}(R^{-1}\tilde{k},\tilde{\lambda})|^2 \frac {w_{0,1}(|\tilde{k}|; k, \lambda)}{|k|^{1/2}} ( \UU_\hh(R)  h)(k , \lambda)  d^3 \tilde{k}
$$
and  justified  dominated convergence by the estimate
\be \label{eq:domconvjust}
| I(\epsilon, k,\lambda,\tilde{\lambda}, R,h ) | \leq    \sup_{r \in I } | w_{0,1}(r ; k, \lambda)|| ( \UU_\hh(R)  h)(k , \lambda)|/\sqrt{|k|}
\ee
and the fact that the r.h.s. of \eqref{eq:domconvjust} is integrable w.r.t. $k$ by  the finiteness of \eqref{eq:defofnorm}.
Using that the square  of \eqref{eq:defoffep} yields a delta sequence
and that $I \ni r  \mapsto w(r; k ,\lambda)$ is for a.e. $k$ a continuous function we find
\begin{equation} \label{eq:resonforint}
\lim_{\epsilon \downarrow 0} I(\epsilon, k,\lambda,\tilde{\lambda}, R,h ) = 2^{-1/2} w_{0,1}( |x| ;  k, \lambda) ( \UU_\hh(R)  h)(k , \lambda)/|k|^{1/2},
 \end{equation}
$k$ a.e. (depending on $R$ and $h$) (notice that we choose  $x$ to lie in the interior
because otherwise we would get a factor $1/2$). Inserting this into  \eqref{dom:conv} we find
\begin{eqnarray}
 \eqref{lhsinvR} =   \sum_\lambda \int
w_{0,1}( |x| ;  k, \lambda) ( \UU_\hh(R)  h)(k , \lambda)\frac{d^3k}{|k|^{1/2}}  .  \label{pt:conv}
\end{eqnarray}
By rotation invariance of $H[w]$ it follows that  \eqref{lhsinvR}  is independent of $R$. In view of  \eqref{pt:conv} this
implies
\begin{equation}
\sum_\lambda \int
w_{0,1}(|x|;  k, \lambda)   h(k , \lambda)\frac{d^3k}{|k|^{1/2}}
 =  \sum_\lambda \int
w_{0,1}(|x|;  k, \lambda) ( \UU_\hh(R)  h)(k , \lambda)\frac{d^3k}{|k|^{1/2}} , \label{eq:invariancerel}
\end{equation}
for all $h \in \hh$,  $R \in G$, and $x \in B_1^0$.
Since  the  the representation  $\UU_\hh(R)$ is unitary and $h \in \hh$ is arbitrary, Eq.  \eqref{eq:invariancerel}
implies  that
\begin{equation} \label{eq:invofwx}
\UU_\hh(R) w^{(|x|)}_{0,1} = w^{(|x|)}_{0,1} , \quad \forall R \in G ,
\end{equation}
where we introduced the function  $w^{(|x|)}_{0,1} : (k,\lambda) \mapsto w_{0,1}(|x|  ;k,\lambda)$
which  by the finiteness of  \eqref{eq:defofnorm}  is an element of $L^2(\R^3 \times \Z_2)$.
Now  \eqref{eq:invofwx} and Corollary   \ref{lem:onepartttt} imply  that   $w^{(|x|)}_{0,1} = 0$.
Since $x \in B_1^0$ is arbitrary, we conclude that $w_{0,1}=0$ (the vanishing at the endpoint $|x|=1$ follows
from continuity) and hence  $H_{0,1}[\underline{w}]=0$. Similarly one shows
$H_{1,0}[\underline{w}]=0$.
\end{proof}

\section{Application}

\subsection{Nonrelativistic qed}
\label{sec:app}

In this subsection we introduce in  \eqref{eq:defofHami} below the Hamiltonian  which
describes, in  the framework of non-relativistic quantum electrodynamics, an
atom consisting of $N$ spin-less electrons and a nucleus with infinite mass and
point charge $Z = N$.
In \eqref{eq:defofatrep} we define a natural $G$--representation on the
Hilbert space of the $N$ electrons, which yields a representation on the
Hilbert space of the total system \eqref{eq:defofreptot}.
In Equation \eqref{eq:trafoG1}   of Proposition  \ref{lem:trafo}   it will be  shown   that, with respect to this
representation, the
Hamiltonian  \eqref{eq:defofHami}     is  a $G$--invariant operator.

The Hilbert space of the $N$ electrons is $\HH_{\rm at} := \bigwedge^N L^2(\R^3)$,
and  the Hilbert space of the total system is  $\HH_0 := \HH_{\rm at} \otimes \FF$.
The Hamiltonian is
\begin{equation} \label{eq:defofHami}
H := \sum_{j=1}^N ( p_j \otimes {\bf 1}  - e  A(x_j) )^2 +  e^2V_C \otimes {\bf 1}   + {\bf 1} \otimes H_f ,
\end{equation}
where $x_j \in \R^3$ is the position of the $j$-th electron, and $e$ is the electron's charge.
$$V_C  := \sum_{i < j} \frac{1}{|x_i - x_j |} - \sum_{j=1}^N \frac{Z}{|x_j|} , $$
with  $l=1,2,3$ and   $x \in \R^3$
\begin{align*}
A(x)   := A^{+}(x) + A^{-}(x) ,  \quad
A_l^{+}(x)  :=  a^*(\kappa_{l,x}) , \quad  A_l^{-}(x) :=  a(\kappa_{l,x})     ,
\end{align*}
 and for $(k,\lambda) \in \R^3 \times \Z_2$
\begin{equation} \label{eq:defofG}
\kappa_{l,x}(k,\lambda) := 1_{|k| \leq \Lambda} \frac{1}{\sqrt{2 |k|}} [ \varepsilon_\lambda(k)]_l  e^{- i k \cdot x }  .
\end{equation}
In \eqref{eq:defofG} the number $\Lambda  > 0$ serves as an ultraviolet cutoff.
The Hamiltonian \eqref{eq:defofHami} can realized as a selfadjoint operator as follows.
One can show that  \eqref{eq:defofHami}
defines  a semibounded closed form on  the natural domain of the
 operator
$(\Delta   \otimes 1 + 1 \otimes H_f )^{1/2}$,  with $\Delta := \sum_{j=1}^n p_j^2$.
By the second representation theorem this yields a unique self-adjoint operator
with domain equal to the natural domain of the operator $\Delta \otimes 1 + 1 \otimes H_f$
(for details see for  example   \cite{HH08}).

We define a $G$--representation, $\UU_{\rm at}$,  on the  Hilbert space $\HH_{\rm at}$ by
setting
\begin{equation} \label{eq:defofatrep}
(\mathcal{U}_{\rm at}(R) \psi)(x_1,...,x_N) := \psi(R^{-1} x_1, ... , R^{-1} x_N) ,
\end{equation}
for all $\psi \in \HH_{\rm at}$, $R \in G$, and $(x_1,...,x_N) \in \R^{3 N}$. On $\HH_0$ we define the tensor representation
\begin{equation} \label{eq:defofreptot}
\UU_0 := \UU_{\rm at} \otimes \UU_{\FF} .
\end{equation}
The next lemma will be used to show  Proposition  \ref{lem:trafo}  below.

\begin{lemma} \label{eq:trafoG}  For all $x \in \R^3$ and $R \in G$ one has   $\UU_\hh(R) \kappa_{l,x} = \sum_{m=1}^3 R_{lm}^{-1} \kappa_{m,R x}$ as an identity in $\hh$.
\end{lemma}
\begin{proof} Using the definitions  \eqref{eq:defofvec}, \eqref{eq:defofphi},  \eqref{eq:defofUhh},  \eqref{eq:defofG},    and the identity
$\sum_{\lambda} [ \varepsilon_\lambda(k) ]_i [\varepsilon_\lambda(k)]_j = \delta_{ij} - \frac{k_i k_j}{|k|^2}$,
one finds
\be
\begin{split}
( \UU_\hh \kappa_{l,x} )(k,\lambda)  = \varepsilon_\lambda(k) \cdot R ( \phi \kappa_{l,x})(R^{-1}k )
 = \varepsilon_\lambda(k) \cdot R \sum_{\lambda'} \varepsilon_{\lambda'}(R^{-1} k) \kappa_{l,x}(\lambda', R^{-1} k )  \cr
 = \sum_{\lambda'}  \varepsilon_\lambda(k) \cdot R \varepsilon_{\lambda'}(R^{-1} k ) \left[ \varepsilon_{\lambda'}(R^{-1} k) \right]_l e^{- i   R^{-1} k \cdot x }
= [ R^{-1} \varepsilon_\lambda(k) ]_l e^{- i  k \cdot R x   } \cr
 =  \sum_{m=1}^3 R_{lm}^{-1} \kappa_{m,R x}(k,\lambda) .
\end{split}
\ee
\end{proof}

Using Lemma  \ref{eq:trafoG} it is now straight forward to prove the following proposition,
which in physicists terminology states that the vector potential transforms as a so called ``vector field'' and that the Hamiltonian
transforms as a so called ``scalar''.

\begin{proposition}
 \label{lem:trafo} The following transformation properties hold.  For all $R \in G$,
\begin{eqnarray}
\UU_0(R)  A^{\pm}(x_j) \UU_0(R)^* & = & R^{-1} A^{\pm}(x_j) , \quad {\rm on \ the  \ natural \ domain \ of } \ (1 \otimes H_f)^{1/2} ,   \label{eq:trafoah} \\
\UU_\FF(R)  H_f \UU_\FF(R)^* & = & H_f  ,  \label{eq:trafoHf} \\
\UU_0(R)  H_g \UU_0(R)^* & = & H_g  . \label{eq:trafoG1}
\end{eqnarray}
\end{proposition}

\begin{proof} Let $\underline{x} = (x_1,...,x_N)$ and $R \underline{x} = (R x_1,..., Rx_N)$.
For $l=1,2,3$ let $x_{j,l}$ denote $l$-th component of $x_j$. Then
$$
( \UU_{\rm at}(R) [x_{j,l} \psi] )(\underline{x}) = [x_{j,l} \psi](R^{-1} \underline{x} ) = [R^{-1} x_j]_l  (\UU_{\rm at}(R) \psi)(\underline{x}) .
$$
This yields the transformation property
\be  \label{eq:trafoxp1}
\UU_{\rm at}(R) x_j \UU_{\rm at}(R)^* = R^{-1} x_j  .
\ee
Let $\partial_{j,l}$ denote the partial derivative with respect to $x_{j,l}$. Then
$$
 ( \UU_{\rm at}(R) [\partial_{j,l} \psi])(\underline{x}) =  [\partial_{j,l} \psi])( R^{-1} \underline{x}) = \sum_{k=1}^3R_{l k}^{-1} \partial_{j,k} [ \UU_{\rm at}(R)  \psi]( \underline{x})  .
$$
This yields the transformation property
\be  \label{eq:trafoxp2}
\UU_{\rm at}(R) p_j \UU_{\rm at}(R)^* = R^{-1} p_j .
\ee
By Lemma \ref{eq:trafoG}, and   \eqref{eq:trafoxp1}, we find
$$
\UU_0(R)  A_l^{\pm}(x_j) \UU_0(R)^* =  a^*( \UU_\hh(R) \kappa_{l,  R^{-1} x_j}) = \sum_{m=1}^3R^{-1}_{lm}  A_m^{\pm}(x_j) .
$$
This implies  \eqref{eq:trafoah}. Eq. \eqref{eq:trafoHf}
can be seen using  the  definition \eqref{eq:defofH} and the fact that
 $\omega$ only depends on $|k|$.
Now \eqref{eq:trafoHf},   \eqref{eq:trafoah},     \eqref{eq:trafoxp1}, and  \eqref{eq:trafoxp2} imply  that  Eq. \eqref{eq:trafoG1} holds a priori in the sense of forms, and hence
as an operator in $\HH_0$.
\end{proof}

For the application of Theorem \ref{thm:main}   in the context of operator theoretic renormalization
we will need Lemma \ref{lem:red}, stated below. To this end, we
first consider  the atomic Hamiltonian
$$H_{\rm at}:=\sum_{j=1}^N p_j^2 +  V_C $$
acting on $\HH_{\rm at}$.
It is straight forward to see that $H_{\rm at}$ is $G$--invariant.
It is well known   that the infimum of the spectrum, $E_{\rm at} := \inf \sigma( H_{\rm at})$,  is an eigenvalue.
We will need the
following rather restrictive Hypothesis.

\medskip

\noindent
{\bf (H)} $E_{\rm at}$ is a   non-degenerate eigenvalue of $H_{\rm at}$.

\medskip

Now suppose   Hypothesis {\bf (H)} holds.
We will denote
by $\varphi_{\rm at}$ the normalized eigenstate of $H_{\rm at}$ with eigenvalue $E_{\rm at}$.
By   hypothesis  it follows that  $\varphi_{\rm at}$ is rotation invariant, since
one dimensional representations of $G$  are trivial. We introduce the
projection $P_{\rm at} := | \varphi_{\rm at} \rangle \langle \varphi_{\rm at} |$ in $\HH_{\rm at}$
and set $P :=  P_{\rm at} \otimes {\bf 1}$.
For later use  we define the map
\begin{equation} \label{defofV}
V_{\rm at}: \FF \to \ran P ,  \quad  \eta \mapsto \varphi_{\rm at} \otimes \eta .
\end{equation}
It follows immediately from  properties of the tensor product that $V_{\rm at}$ is   unitary.

\begin{lemma} \label{lem:red} Suppose {\bf (H)} holds. Let $T$ be a bounded  $G$--invariant operator on $\HH_0$.  Then $V_{\rm at}^* P T P V_{\rm at}$ is
 a $G$--invariant operator on $\FF$.
\end{lemma}

\begin{proof}
For any $\eta_1, \eta_2 \in \FF$ and $R \in G$,
\be
\begin{split}
\big\langle\eta_1 ,  V_{\rm at}^*  P  T P  V_{\rm at} \eta_2 \big\rangle   =    \big\langle \varphi_{\rm at} \otimes \eta_1     ,   T  \ [ \varphi_{\rm at} \otimes  \eta_2  ] \big\rangle =
\big\langle \varphi_{\rm at}  \otimes  \eta_1  , \UU_0(R) \  T \ \UU_0(R)^* \ [   \varphi_{\rm at} \otimes \eta_2  ] \big\rangle \nonumber \cr
= \big\langle   \varphi_{\rm at} \otimes  ( \UU_\FF(R)^*  \eta_1 ) ,    T  \  [    \varphi_{\rm at}   \otimes  (\UU_\FF(R)^* \eta_2)  ]  \big\rangle  =
\big\langle\UU_\FF(R)^*  \eta_1 ,  V_{\rm at}^*  P  T P  V_{\rm at}  \UU_\FF(R)^* \eta_2 \big\rangle \nonumber\cr =
\big\langle \eta_1 , \UU_\FF(R)  V_{\rm at}^*  P  T P  V_{\rm at}  \UU_\FF(R)^* \eta_2 \big\rangle
\end{split}
\ee
where we used the rotation invariance of $\varphi_{\rm at}$. Since $\eta_1, \eta_2$ are arbitrary, the Lemma follows.
\end{proof}

\subsection{Smooth Feshbach}
\label{sec:smo}

In this subsection we first introduce  the so called Feshbach operator, and then
state  in Corollary \ref{cor:main}, below, the main application of  Theorem \ref{thm:main}.
This Corollary can be used in operator theoretic renormalization to show that marginal terms are absent \cite{HH10-2}.

Let $\chi$ and $\overline{\chi}$ be commuting, nonzero bounded operators, acting on a separable Hilbert space $\HH$
and satisfying $\chi^2 + \overline{\chi}^2=1$.
A {\it Feshbach pair} $(H,T)$ for $\chi$ is a pair of
closed operators with the same domain,
$$
H,T : D(H) = D(T) \subset \HH \to \HH
$$
such that $H,T, W := H-T$, and the operators
\begin{align*}
&W_\chi := \chi W \chi , & &W_{\overline{\chi}} := \overline{\chi} W \chib \\
&H_\chi :=T + W_\chi , & &H_{\overline{\chi}} := T + W_{\chib} ,
\end{align*}
defined on $D(T)$ satisfy the following assumptions:
\begin{itemize}
\item[(a)] $\chi T \subset T \chi$ and $\chib T \subset T \chib$,
\item[(b)] $H_{\chib}, T  : D(T) \cap \ran \chib \to \ran \chib$ are  bijections with bounded inverse,
\item[(c)] $\chib H_{\chib}^{-1} \chib W \chi : D(T) \subset \HH \to \HH$ is a bounded operator.
\end{itemize}
Here we used the notation   $H_{\chib}^{-1} \chib :=  \left( H_{\chib} \upharpoonright \ran \chib \right)^{-1} \chib$.
Given a Feshbach pair $(H,T)$ for $\chi$, the operator
\begin{align} \label{eq:defoffesh}
&F_\chi(H,T) := H_\chi - \chi W \chib H_{\chib}^{-1} \chib W \chi
\end{align}
is called Feshbach operator.
In  \cite{GH08} it is shown that the full spectral information of the original operator $H$
can be recovered by the  restriction of  the Feshbach operator to  any closed subspace  $V$ with the property that
$\ran \chi \subset V \subset \HH$ and that $\chi V \subset V $.

In operator theoretic renormalization one typically chooses  the following
realization for the operators $\chi$ and $\overline{\chi}$.
Let $\eta  \in C([0,1];[0,1])$ be a function such that there exist two numbers $a$ and $b$ with  $0 < a < b  <1$,
$\eta|_{[0,a]}  = 1$,  and $\eta|_{[b,1]} = 0$.
Setting  $\overline{\eta} := (1 - \eta^2)^{1/2}$ it follows that $\eta^2 + \overline{\eta}^2 =1$.
The operators $\boldsymbol{\chi} := P_{\rm at}  \otimes \eta(H_f)$  and $\boldsymbol{\overline{\chi}}  := ( {\bf 1} - P_{\rm at} )\otimes {\bf 1} + P_{\rm at} \otimes \overline{\eta}(H_f)$
satisfy
$$\boldsymbol{\chi}^2 + \boldsymbol{\overline{\chi}}^2=1 .$$

\begin{corollary}\label{cor:main}  The following statements hold.
\begin{itemize}
\item[(a)]
Fix $z,g \in \C$. Suppose Hypothesis {\bf (H)} holds and that  $(H_g - z, H_0 - z)$ is a Feshbach pair for  $\boldsymbol{\chi}$.
Then the operator $F_{z,g} := V_{\rm at}^* P  F_{\boldsymbol{\chi}}(H_g - z,H_0-z)  P  V_{\rm at}$ is invariant under rotations.
If there exists a $\underline{w}_{z,g} \in \WW_\xi$ such that $F_{z,g} = H[\underline{w}_{z,g}]$,
then $H_{0,1}[\underline{w}_{z,g}]=H_{0,1}[\underline{w}_{z,g}]=0$.
\item[(b)]  Suppose $\chi  = \eta(H_f)$ and  $(H, T)$ is a Feshbach pair for  $\chi$.
Assume that $H$ and $T$ are rotation invariant operators on $\FF$.   Then
$F_\chi(H,T)$ is invariant under rotations. If there exists a $\underline{w} \in \WW_\xi$ such
that  $F_\chi(H,T)=H[\underline{w}]$, then   $H_{0,1}[\underline{w}]=H_{0,1}[\underline{w}]=0$.
\end{itemize}
\end{corollary}
\begin{proof} (a). From Lemma  \ref{lem:trafo} and the definition \eqref{eq:defoffesh}
it follows that $F_{\boldsymbol{\chi}}(H_g - z,H_0-z)$ is rotation invariant. Thus
the rotation invariance of $F_{z,g}$  now follows from  Lemma \ref{lem:red}.
Hence by Theorem \ref{thm:main}, $H_{0,1}[\underline{w}_{z,g}]=H_{0,1}[\underline{w}_{z,g}]=0$.
(b). From Lemma  \ref{lem:trafo} and the definition \eqref{eq:defoffesh}
it follows that $F_\chi(H,T)$ is rotation invariant.
Hence by Theorem \ref{thm:main}, $H_{0,1}[\underline{w}]=H_{0,1}[\underline{w}]=0$.
\end{proof}

\begin{remark}
Part~(a) of the corollary  is used for a so called  initial   Feshbach
operator and Part~(b) is used for each renormalization step.
Corollary \ref{cor:main}  is stated under the assumptions that
the Feshbach operator can be expressed  in terms of
integral kernels  \eqref{defofintop0}.  In  \cite{HH10-2} this assumption
is verified for the initial step, see also \cite{BFS98}.  For the renormalization step
this property  is  shown to hold under the natural assumptions needed  for  operator theoretic renormalization, \cite{BCFS03}.
\end{remark}



\appendix

\section{Appendix}

In this appendix we give a rigorous definition of \eqref{defofintop}, which does not
involve creation or annihilation operators,
and  we provide a proof of  \eqref{eq:boundonH}.
We introduce the  Hilbert space
$
L^2_s( \left\{ \R^3 \times \Z_2 \right\}^n )$
of complex valued square integrable functions 
which are  totally   symmetric with  respect to the interchange of   arguments belonging to
different  factors  of the $n$-fold Cartesian product.
We will identify this Hilbert space with a subspace of Fock space,
by means of the canonical isomorphism of Hilbert spaces
$$
\FF_n \cong L^2_s( \left\{ \R^3 \times \Z_2 \right\}^n ) .
$$
Let $\varphi_p \in \FF_p$ and $\psi_q \in \FF_q$.
If
\begin{equation}\label{eq:cond22}
 p - m = q - n \geq 0 ,
\end{equation}
we define
\begin{eqnarray} \label{eq:defofL2s}
  \langle \varphi_p , H_{m,n}[\underline{w}] \psi_q \rangle  &:=    &
   \sqrt{\frac{ q !   }{l ! }}   \sqrt{\frac{  p !  }{l ! }}
\int_{X^m \times X^l \times X^n }  d \hat{K}^{(l)}  \frac{d \tilde{K}^{(m)}}{| \tilde{K}^{(m)} |^{1/2}} \frac{ d {K}^{(n)}}{ |{K}^{(n)}|^{1/2}}  \\
&& \times
\overline{ \varphi_p(\tilde{K}^{(m)}  ,  \hat{K}^{(l)}  ) }
w_{m,n}( \Sigma[\hat{K}^{(l)}] ;  \tilde{K}^{(m)} ,  {K}^{(n)})
 \psi_q({K}^{(n)} , \hat{K}^{(l)} )
 , \nonumber
\end{eqnarray}
where we used the following definitions,  $l :=  p -  m =  q-n$.
If  \eqref{eq:cond22} does not hold, we define
\begin{equation} \label{eq:zero}
\langle \varphi_p , H_{m,n}[\underline{w}] \psi_q \rangle  := 0 .
\end{equation}
 The goal
of the remaining part of the Appendix is to  show \eqref{eq:boundonH}.
It will then follow from \eqref{eq:boundonH}  and the  Riesz representation theorem, that
$H_{m,n}[\underline{w}]$ defines a bounded operator on Fock space.
In the case of   \eqref{eq:cond22}  we estimate using Cauchy-Schwarz
\begin{eqnarray}
\lefteqn{  |  \langle \varphi_p , H_{m,n}[\underline{w}] \psi_q \rangle  |  \leq  } \nonumber  \\
&&    D_{l,p}(\varphi_p) D_{l,q}(\psi_q)
 \left\{ \int_{X^m \times X^n}  | {K}^{(n)}|^{-2}  |\tilde{K}^{(m)}|^{-2}
\sup_{r \in I} |  w_{m,n}( r ;  \tilde{K}^{(m)} ,  {K}^{(n)} ) |^2  \right\}^{1/2} , \label{eq:defofL2svv}
\end{eqnarray}
where we defined
\begin{eqnarray*}
 D_{l,p}(\varphi_p) :=  \sqrt{\frac{  p !  }{l ! }}
\left\{ \int_{ X^m \times X^l }    d \tilde{K}^{(m)}  d \hat{K}^{(l)}
|\tilde{K}^{(m)}| \left|  \varphi_p(\tilde{K}^{(m)}  ,  \hat{K}^{(l)}  )
P_{\rm red}(\tilde{K}^{(m)}, \hat{K}^{(l)})\right|^2 \right\}^{1/2}  ,
\end{eqnarray*}
and inserted   $P_{\rm red}(\tilde{K}^{(m)}, \hat{K}^{(l)})
:=   1_{\Sigma[\tilde{K}^{(m)}] + \Sigma[\hat{K}^{(l)}] \leq 1}$ justified by  \eqref{eq:supportprop}.
Now using the symmetry property of the wavefunctions, one finds
\begin{eqnarray*}
 D_{l,p}(\varphi_p) && =  \sqrt{\frac{  p !  }{l ! }}
\Big\{ \int_{ X^1 \times X^{m-1}  \times X^l }  d K^{(1)}  d \tilde{K}^{(m-1)}  d \hat{K}^{(l)}
| {K}^{(1)}|    |\tilde{K}^{(m-1)}| \\
&& \times \left|  \varphi_p({K}^{(1)}, \tilde{K}^{(m-1)}  ,  \hat{K}^{(l)}  )
P_{\rm red}({K}^{(1)}, \tilde{K}^{(m-1)}, \hat{K}^{(l)})\right|^2 \Big\}^{1/2}  \\
&&=    \sqrt{\frac{  p !  }{( l+ 1) ! }}
\Big\{ \int_{  X^{m-1}  \times X^{l+1} }    d \tilde{K}^{(m-1)}  d \hat{K}^{(l+1)}
  |\tilde{K}^{(m-1)}| \\
&& \times \left| ( \Sigma[\hat{K}^{(l+1)}])^{1/2} \varphi_p( \tilde{K}^{(m-1)}  ,  \hat{K}^{(l+1)}  )
P_{\rm red}( \tilde{K}^{(m-1)} , \tilde{K}^{(l +1)} )\right|^2 \Big\}^{1/2}  \\
&&\leq     \sqrt{\frac{  p !  }{( l+ 1) ! }}
\Big\{ \int_{  X^{m-1}  \times X^{l+1} }    d \tilde{K}^{(m-1)}  d \hat{K}^{(l+1)}
  |\tilde{K}^{(m-1)}| \\
&& \times \left| (H_f)^{1/2} \varphi_p( \tilde{K}^{(m-1)}  ,  \hat{K}^{(l+1)}  )
P_{\rm red}( \tilde{K}^{(m-1)} , \tilde{K}^{(l +1)} )\right|^2 \Big\}^{1/2} .
\end{eqnarray*}
 Iterating
above estimate  we arrive at
$
D_{l,p}(\varphi_p) \leq  \| H_f^{m/2} 1_{H_f \leq 1} \varphi_p \| $.
Inserting this into  \eqref{eq:defofL2svv} gives
$$
| \langle   \varphi_p , H_{m,n}[\underline{w}]  \psi_q \rangle
| \leq \| w_{m,n} \|_{\WW_{m,n}} \| \varphi_p \| \| \psi_q \| .
$$
This  together with \eqref{eq:zero} yields that  for any two vectors  $\varphi = (\varphi_n)_{n \in \N_0}$ and
$\psi =   (\psi_n)_{n \in \N_0}$  in  $\FF$  (i.e., $\varphi_n , \psi_n \in \FF_n$) the following
inequality holds,
$$
| \langle \varphi , H_{m,n}[\underline{w}] \psi \rangle | \leq \| w_{m,n} \|_{\WW_{m,n}} \sum_{l=0}^\infty \| \varphi_{l+m} \| \| \psi_{l+n} \| \leq  \| w_{m,n} \|_{\WW_{m,n}} \| \varphi \| \| \psi \| ,
$$
and hence \eqref{eq:boundonH} follows.

\end{document}